\documentclass{article}
\usepackage{spconf,amsmath,graphicx,bbm,array,dcolumn,epsfig,amsmath,amsfonts,yhmath,bm,amssymb,psfrag,color,soul,enumerate,stackengine,cite}
\usepackage[dvipsnames]{xcolor}
\usepackage{tcolorbox}
\usepackage[font=footnotesize,skip=7pt]{caption}
\usepackage[font=small,skip=0pt]{subcaption}

\usepackage{tikz}

\definecolor{lgreen} {RGB}{180,210,100}
\definecolor{ngreen} {RGB}{98,158,31}
\definecolor{dgreen} {RGB}{78,138,21}
\definecolor{MLOWLSgreen} {RGB}{0,140,130}
\definecolor{SDPpurple} {RGB}{191,0,191}
\definecolor{lred}   {RGB}{220,0,0}
\definecolor{nred}   {RGB}{224,0,0}
\definecolor{bred}   {RGB}{200,20,20}
\definecolor{nblue}  {RGB}{28,130,185}
\definecolor{jblue}  {RGB}{20,50,100}
	
\definecolor{bred}   {RGB}{200,20,20}
\definecolor{crimson} {RGB}{220,20,62}
\newcommand*\circled[1]{\tikz[baseline=(char.base)]{
		\node[circle,draw,color=crimson, opacity=0.9,inner sep=1pt] (char) {\footnotesize #1};}}
\newcommand*\circledsmall[1]{\tikz[baseline=(char.base)]{
		\node[circle,draw,color=crimson, opacity=0.9,inner sep=1pt] (char) {\scriptsize #1};}}

\newcommand*{\myfontb}{\fontfamily{lmr}\selectfont}

\newcommand {\myvec}[1] {{\mbox{\boldmath $#1$}}}
\newcommand {\mymat}[1]  {{\mbox{\boldmath $#1$}}}

\DeclareMathAlphabet      {\mathbfit}{OML}{cmm}{b}{it}

\AtBeginEnvironment{bmatrix}{\setlength{\arraycolsep}{4pt}}

\newcommand {\U} {\mymat{U}}

\newcommand {\mLambda} {\mymat{\Lambda}}

\newcommand {\C} {\mymat{C}}
\newcommand {\D} {\mymat{D}}

\newcommand {\R} {\mymat{R}}
\newcommand {\hR} {\widehat{\R}}
\newcommand {\bR} {\mybar{\R}}
\newcommand {\I} {\mymat{I}}
\newcommand {\X} {\mymat{X}}

\newcommand {\ua} {\myvec{a}}
\newcommand {\ub} {\myvec{b}}

\newcommand {\up} {\myvec{p}}

\newcommand {\uxi} {\myvec{\xi}}

\newcommand {\uv} {\myvec{v}}

\newcommand {\us} {\myvec{s}}

\newcommand {\ur} {\myvec{r}}

\newcommand {\uy} {\myvec{y}}

\newcommand {\uz} {\myvec{z}}

\newcommand {\Rset} {\mathbb{R}}
\newcommand {\Cset} {\mathbb{C}}
\newcommand {\Zset} {\mathbb{Z}}
\newcommand {\Eset} {\mathbb{E}}

\newcommand {\tps} {\tiny \rm{T}}
\newcommand {\her} {\tiny \rm{H}}

\newcommand {\calq} {\mathcal{Q}}
\newcommand {\sgn} {\mbox{sign}}
\newcommand {\sine} {\mbox{sine}}
\newcommand {\bur} {\mybar{\ur}}
\newcommand {\buv} {\mybar{\uv}}
\newcommand {\bus} {\mybar{\us}}
\newcommand {\buy} {\mybar{\uy}}
\newcommand {\sbar} {\mybar{s}}
\newcommand {\hup} {\widehat{\up}}
\newcommand {\dpd} {\text{\tiny DPD}}
\newcommand {\qdpd} {\text{\tiny OB-DPD}}
\DeclareMathOperator*{\argmin}{argmin}
\DeclareMathOperator*{\argmax}{argmax}

\newcommand {\ud} {\myvec{d}}
\newcommand {\bD} {\mybar{\D}}

\newtheorem{proof}{Proof}
\newtheorem{corollary}{Corollary}

\makeatletter
\newsavebox\myboxA
\newsavebox\myboxB
\newlength\mylenA

\newcommand*\mybar[2][0.75]{%
	\sbox{\myboxA}{$\m@th#2$}%
	\setbox\myboxB\null
	\ht\myboxB=\ht\myboxA%
	\dp\myboxB=\dp\myboxA%
	\wd\myboxB=#1\wd\myboxA
	\sbox\myboxB{$\m@th\overline{\copy\myboxB}$}
	\setlength\mylenA{\the\wd\myboxA}
	\addtolength\mylenA{-\the\wd\myboxB}%
	\ifdim\wd\myboxB<\wd\myboxA%
	\rlap{\hskip 0.5\mylenA\usebox\myboxB}{\usebox\myboxA}%
	\else
	\hskip -0.5\mylenA\rlap{\usebox\myboxA}{\hskip 0.5\mylenA\usebox\myboxB}%
	\fi}
\makeatother

\title{One-Bit Direct Position Determination of Narrowband Gaussian Signals}

\name{Amir Weiss and Gregory W. Wornell}

\address{
\begin{tabular}{cc}
Department of Electrical Engineering and Computer Science\\
Massachusetts Institute of Technology\\
\{amirwei,gww\}@mit.edu
\end{tabular}
}

\begin{document}
\ninept
\maketitle
\setlength{\abovedisplayskip}{5pt}
\setlength{\belowdisplayskip}{5pt}

\begin{abstract}
	One of the main drawbacks of the well-known Direct Position Determination (DPD) method is the requirement that raw data be transferred to a common processor. It would therefore be of high practical value if DPD---or a modified version thereof---could be successfully applied to a coarsely quantized version of the raw data, thus alleviating the communication requirements between the different base stations. Motivated by the above, and inspired by recent work in the rejuvenated one-bit array processing field, we present One-Bit DPD: a direct localization method based on one-bit quantized measurements. We show that despite the coarse quantization, the proposed method nonetheless yields a position estimate with appealing asymptotic properties. We further establish the identifiability conditions of this model, which rely only on second-order statistics. Simulation results corroborate our analytical derivations, demonstrating that much of the information regarding the emitter position is preserved under this crude form of quantization.
	%
\end{abstract}

\vspace{-0.3cm}
\section{Introduction}\label{sec:intro}
\vspace{-0.05cm}
Emitter localization is one of the most attractive problems in signal processing, relevant to a variety of applications, such as autonomous vehicles \cite{caiti2005localization,ward2016vehicle}, radar, lidar and sonar systems \cite{becker1999passive,wolcott2017robust,liu2016robotic}, and battlefield surveillance \cite{wan2015distributed}, to name but a few.

Naturally, this problem has been extensively addressed in the literature and different methods have been proposed over the years. Common two-step methods are based on first estimating the signal's Time-Of-Arrival (TOA) and / or Angle-Of-Arrival (AOA) at several base stations, and based on these estimates, the emitter position is subsequently estimated. Other methods, operating by the same principle, incorporate Frequency Difference-Of-Arrival (FDOA) estimation as well \cite{ho2004accurate,ho2007source}, when possible. Since these TOA / AOA / FDOA are computed separately at each base station, they do not account for the informative constraint that all measurements from all the different base stations contain the {\myfontb\emph{same}} signal, transmitted from the underlying emitter at the unknown location.

In order to further exploit the ``extra" information within this constraint, the one-step Direct Position Determination (DPD) method was proposed by Weiss in his seminal work \cite{weiss2004direct}, where DPD of a single emitter is addressed. Following, many variants and extensions have been proposed, e.g., \cite{weiss2005direct,miljko2008direct,oispuu2010direct,tirer2015high,keskin2017direct,lu2017novel,zhang2018direct,ma2020distributed}. 

Although DPD provides higher accuracy than two-step-based alternatives, it requires the transmission of raw signal data from all base stations to a central computing unit. This entails a trade-off between localization accuracy and communication resources, which is at the heart of this paper. Our goal is thus to take one step towards a better balance with respect to (w.r.t.) this trade-off. In other words, for a prespecified desired localization accuracy, we aim to ``strip" significant redundancies from the transmitted raw data.

Inspired by recent work from the flourishing field of one-bit processing (e.g., \cite{choi2016near,gao2017gridless,ren2017one,liu2017one,ameri2019one,huang2019one,cheng2019target,sedighi2020one}), our main contribution in this letter is One-Bit DPD (OB-DPD)---a method allowing for accurate direct localization based on quantized measurements of Gaussian signals, acquired by a $1$-bit Analog-to-Digital Converter (ADC) low-complexity receiver. In addition, we establish the respective identifiability conditions for the quantized signal model under consideration, based only on Second-Order Statistics (SOS).

\vspace{-0.35cm}
\section{Problem Formulation}\label{sec:problemformulation}
\vspace{-0.1cm}
Consider $L$ spatially diverse base stations, each equipped with an $M$-element calibrated omni-directional antenna array, and the presence of an unknown narrowband signal, emitted from a transmitter whose deterministic, unknown position is denoted by the vector of coordinates\footnote{$K_p\hspace{-0.05cm}=2$ for the $2$-dimensional case; $K_p=3$ for the $3$-dimensional case.} $\up\in\Rset^{K_p\times 1}$. We assume that the transmitter is sufficiently far from all $L$ base stations to allow a planar wavefront (``far-field") approximation. The time-varying vector of sampled, baseband-converted signals from all $M$ sensors at the $\ell$-th base station array is given by
\begin{equation}\label{modelequation}
	\ur_{\ell}[n]=b_{\ell}\ua_{\ell}(\up)s_{\ell}[n]+\uv_{\ell}[n]\in\Cset^{M\times 1},
\end{equation}
for all time-index $n\in\{1,\ldots,N\}$, where
\begin{enumerate}[i]
	\itemsep0.025em 
	\item $b_{\ell}\in\Cset$ denotes the unknown channel effect (e.g., attenuation) from the transmitter to the $\ell$-th base station;
	\item $\ua_{\ell}(\up)\in\Cset^{M\times 1}$ denotes the known $\ell$-th array response to a signal transmitted from position $\up$; 
	\item $s_{\ell}[n]\triangleq \left.s\left(t-\tau_{\ell}(\up)\right)\right\vert_{t=nT_s}\in\Cset$ denotes the unobservable sampled signal waveform at the $\ell$-th base station, where $s\left(t-\tau_{\ell}(\up)\right)$ is the analog, continuous-time waveform delayed by $\tau_{\ell}(\up)$, and $T_s$ is the sampling period; and
	\item $\uv_{\ell}[n]\in\Cset^{M\times1}$ denotes the additive noise at the $\ell$-th base station, representing internal (e.g., thermal) receiver noise and ``interfering" signals, modeled as spatially and temporally independent, identically distributed (i.i.d.) zero-mean circular Complex Normal (CN, \cite{loesch2013cramer}) vector process with a covariance matrix $\R_{v_{\ell}}\triangleq\Eset\left[\uv_{\ell}[n]\uv^{\her}_{\ell}[n]\right]=\sigma_{v_{\ell}}^2\I_M\in\Cset^{M\times M}$, where $\sigma_{v_{\ell}}^2$ is unknown, and $\I_M$ denotes the $M$-dimensional identity matrix.
\end{enumerate}

We also assume that the transmitter and antenna arrays are stationary during the whole observation interval, such that no frequency shift due to the Doppler effect occurs. Further, we assume that $s(t)$ may be modeled as a stationary (not necessarily i.i.d.) zero-mean circular CN random process, statistically independent of all $\{\uv_{\ell}[n]\}_{\ell=1}^L$, with an unknown auto-correlation function $R_s(t_0)\triangleq\Eset\left[s(t)s^*(t-t_0)\right]\in\Cset$. Consequently, we get for all $i,j\in\{1,\ldots,L\}$,
\begin{equation*}\label{crosscorrelationofsignal}
	R^{(i,j)}_s[m]\triangleq\Eset\left[s_{i}[n]s_{j}^*[n-m]\right]=R_s\left(mT_s-\Delta_{ij}(\up)\right),
\end{equation*}
where $\Delta_{ij}(\up)\triangleq\tau_{i}(\up)-\tau_{j}(\up)\in\Rset$. Furthermore, it follows that $\{\ur_{\ell}[n]\}_{\ell=1}^L$ are all stationary zero-mean circular CN with auto- and cross-covariance matrices (for all $i,j\in\{1,\ldots,L\}$)
\begin{equation*}\label{autoandcorrscovariance}
	\begin{gathered}
		\R^{(i,j)}[m]\triangleq\Eset\left[\ur_{i}[n]\ur^{\her}_{j}[n-m]\right]=\\
		b_{i}b^*_{j}R^{(i,j)}_s[m]\ua_{i}(\up)\ua^{\her}_{j}(\up)+\delta_{ij}\sigma^2_{v_{i}}\I_M\in\Cset^{M\times M},
	\end{gathered}
\end{equation*}
where $\delta_{mn}$ denotes the Kronecker delta of $m,n\in\Zset$. Lastly, applying the (normalized) Discrete Fourier Transform (DFT) to \eqref{modelequation} yields the equivalent frequency-domain representation,
\begin{equation*}\label{modelequationfreq}
	\begin{gathered}
		\bur_{\ell}[k]=b_{\ell}\ua_{\ell}(\up)\sbar[k]e^{-\jmath\omega_k\tau_{\ell}({\text{\boldmath $p$}})}+\buv_{\ell}[k]\in\Cset^{M\times 1},\\
		\omega_k\triangleq\frac{2\pi(k-1)}{NT_s}\in\Rset^+, \quad \forall k\in\{1,\ldots,N\},
	\end{gathered}
\end{equation*}
where we use $\mybar{z}[k]$ to denote the $k$-th DFT coefficient of the corresponding (discrete) time-domain signal $z[n]$ everywhere.

In this work, rather than assuming access to the discrete-time signal \eqref{modelequation} measured by an ideal $\infty$-bit ADC, we assume access only to a ``coarse" quantized version thereof, obtained by a $1$-bit ADC. Specifically, the vector of $1$-bit quantized signals from the $\ell$-th base station at time $n$ is given by
\begin{equation}\label{onebitmeasuredsignal}
	\uy_{\ell}[n]\triangleq\calq\big(\ur_{\ell}[n]\big)\in\left\{e^{\jmath\left(\frac{\pi}{4}+\frac{\pi m}{2}\right)}: 0\leq m\leq3 \right\}^{M\times 1},
\end{equation}
where the complex-valued $1$-bit quantizer is defined as
\begin{equation}\label{defofQ}
	\mathcal{Q}(z)\triangleq\frac{1}{\sqrt{2}}\cdot\Big[\sgn\left(\Re\{z\}\right)+\jmath\cdot\sgn\left(\Im\{z\}\right)\Big], \forall z\in\Cset,
\end{equation}
and, with slight abuse of notations, $\calq(\cdot)$ operates elementwise in \eqref{onebitmeasuredsignal}. For simplicity of the exposition in some parts of the derivation which follows, we further assume that all $\{\tau_{\ell}(\up)\}$ are (at least approximately) an integer multiple of the sampling period. However, this assumption is not required in practice, and can be relaxed. This completes the definition of our model, and the problem at hand may be stated concisely as follows:
\tcbset{colframe=gray!90!blue,size=small,width=0.49\textwidth,arc=2mm,outer arc=1mm}
\begin{tcolorbox}[upperbox=visible,colback=white]
	{\emph{Given the $1$-bit measurements $\{\uy_{\ell}[n]\}$ from all $L$ base stations, estimate the unknown transmitter position $\up$.}}
\end{tcolorbox}
\vspace{-0.5cm}
\section{DPD for Non-Quantized Measurements}\label{sec:DPDreview}
Since our proposed algorithm is both inspired by and closely related to the original DPD method \cite{weiss2004direct}, we first briefly review the DPD objective function, principal steps of the derivation, and the final closed-form expression for the estimate of $\up$.

The DPD method is in fact originated by Nonlinear Least-Squares (NLS) estimation of the transmitter's position $\up$ based on the data $\{\ur_{\ell}[n]\}$. Thus, it seeks the position estimate\footnote{For brevity, whenever it is clear from the context, we hereafter loosely use $\up$ for the true emitter position, and for a general position vector-variable.}
\begin{equation*}\label{originalDPDdefinition}
	\hup_{\dpd}\triangleq\underset{{\text{\boldmath $p$}}\in\Rset^{K_{p}\times 1}}{\argmin}\;Q(\up),
\end{equation*}
where the NLS cost function to be minimized is given by
\begin{equation}\label{originalDPDcostfunction}
	Q(\up)\triangleq \min_{\substack{\text{\boldmath$\bar{s}$}\in\Cset^{N\times1} \\ {\text{\boldmath$b$}}\in\mathcal{B}_L}} \sum_{\ell=1}^{L}\left\|\bur_{\ell}-b_{\ell}\cdot\bus_{\ell}\otimes\ua_{\ell}(\up)\right\|^2\in\Rset^+,
\end{equation}
$\|\cdot\|$ and $\otimes$ are the $\ell^2$-norm and Kronecker product, resp., $\mathcal{B}_L\triangleq\{\uz\in\Cset^{L\times1}:\|\uz\|=1\}$, $\bus\triangleq\left[\sbar[1] \cdots \sbar[N]\right]^{\tps}\in\Cset^{N\times 1}$, $\ub\triangleq\left[b_1 \cdots b_{L}\right]^{\tps}\in\Cset^{L\times 1}$, and we have further defined
\begin{align*}\label{notation1}
	&\bur_{\ell}\triangleq\left[\bur^{\tps}_{\ell}[1] \cdots \bur^{\tps}_{\ell}[N]\right]^{\tps}\in\Cset^{MN\times 1},\\
	&\bus_{\ell}\triangleq\left[\sbar[1]e^{-\jmath\omega_1\tau_{\ell}({\text{\boldmath $p$}})} \cdots \sbar[N]e^{-\jmath\omega_N\tau_{\ell}({\text{\boldmath $p$}})}\right]^{\tps}\in\Cset^{N\times 1},
\end{align*}
for all possible $\ell$. Note that since both $\bus$ and $\ub$ are unknown, assuming $\ub\in\mathcal{B}_L$ is without loss of generality (see, e.g., \cite{tirer2015high}).

As shown in \cite{weiss2004direct}, optimizing \eqref{originalDPDcostfunction} w.r.t.\ the channel parameters $\ub$, and then further optimizing w.r.t.\ the signal's DFT coefficients $\bus$, yields the compact and elegant from,
\begin{equation}\label{compatformofDPD}
	\hup_{\dpd}=\underset{{\text{\boldmath $p$}}\in\Rset^{K_{p}\times 1}}{\argmax}\;\lambda_{\text{max}}\left(\D(\up)\right).
\end{equation}
Here, $\lambda_{\text{max}}(\X)$ denotes the largest eigenvalue of its square-matrix argument $\X$, and the matrix $\D(\up)$ is defined as
\begin{equation}\label{Ddef}
	\D(\up)\triangleq\U^{\her}(\up)\U(\up)\in\Cset^{L\times L},
\end{equation}
where $\U(\up)\triangleq\left[\ud_1 \cdots \ud_L\right]\in\Cset^{N\times L}$, and
\begin{align*}
	\ud_{\ell}&\triangleq\left[d_{\ell}[1] \cdots d_{\ell}[N]\right]^{\tps}\in\Cset^{N\times 1}, \;\forall \ell\in\{1,\ldots,L\},\\
	d_{\ell}[k]&\triangleq e^{-\jmath\omega_k\tau_{\ell}({\text{\boldmath $p$}})}\bur^{\her}_{\ell}[k]\ua_{\ell}(\up)\in\Cset, \;\forall k\in\{1,\ldots,N\}.
\end{align*}
As pointed out, e.g., in \cite{tirer2015high}, when $\bus$ is considered deterministic unknown and $\sigma^2_{v_{\ell}}$ are all equal (and known), \eqref{compatformofDPD} is also the maximum likelihood estimate of $\up$. Nevertheless, it is always the NLS estimate of $\up$, {\myfontb\emph{regardless}} of the noise characteristics and / or the underlying nature of the unknown signal $s(t)$, be it deterministic or random. In particular, \eqref{compatformofDPD} is the NLS estimate within our random CN signal model as well. This observation will be used in Section \ref{sec:OneBitDPDsolution} in order to characterize some of the asymptotic properties of our proposed OB-DPD estimate. We note that a detailed analysis under this particularly interesting CN signal model is given in \cite{weiss2005direct}, Appendix C.

Next, we provide an intuitive interpretation of the solution \eqref{compatformofDPD}, which is instrumental for the subsequent derivation of the OB-DPD algorithm for quantized $1$-bit measurements.
\subsection{Joint Beamforming: DPD as a Collection of Beamformers}\label{subsec:DPDbeamformer}
Careful inspection of the element $D_{ij}(\up)$ of \eqref{Ddef} reveals that\vspace{-0.1cm}
\begin{equation}\label{dataonlyviaestimate}
	\begin{aligned}
		D_{ij}(\up)&=\ud_i^{\her}\ud_j=\sum_{k=1}^{N}d_i^*[k]d_j[k]=\\
		&=\ua^{\her}_{i}(\up)\Bigg(\underbrace{\sum_{k=1}^{N}\bur_{i}[k]\bur^{\her}_{j}[k]e^{\jmath\omega_k\Delta_{ij}({\text{\boldmath $p$}})}}_{\triangleq N\cdot{\text{\boldmath $\widehat{R}$}^{(i,j)}}[\Delta_{ij}({\text{\boldmath $p$}})]\in\Cset^{M\times M}}\Bigg)\ua_{j}(\up)\\
		&=N\cdot\ua^{\her}_{i}(\up)\hR^{(i,j)}[\Delta_{ij}(\up)]\ua_{j}(\up).
	\end{aligned}
\end{equation}
Hence, $D_{ij}(\up)$ can be viewed as conventional (Bartlett, \cite{van2002optimum}) ``cross"-beamforming between the observed signals at the $i$-th and $j$-th base stations. This becomes even clearer when considering the ``large" sample size asymptotic regime, as\vspace{-0.2cm}
\begin{equation*}\label{asymptoticconvergence}
	\begin{aligned}
		&\hR^{(i,j)}[\Delta_{ij}(\up)]=\frac{1}{N}\sum_{k=1}^{N}\bur_{i}[k]\bur^{\her}_{j}[k]e^{\jmath\omega_k\Delta_{ij}({\text{\boldmath $p$}})}\underset{\circledsmall{$1$}}{=}\\
		\vspace{-0.2cm}&\frac{1}{N}\sum_{\substack{n=1 \\ \,}}^{N}\ur_{i}[n+\tau_i(\up)]\ur_{j}^{\her}[n+\tau_j(\up)]\xrightarrow[\;\;\circledsmall{$2$}\;\;]{P}\R^{(i,j)}[\Delta_{ij}(\up)],
	\end{aligned}
\end{equation*}
where $\xrightarrow[\;\;\;\;]{P}$ denotes convergence in probability as $N\rightarrow\infty$ \cite{evans2004probability}, and we have used in \circled{$1$} Parseval's theorem\footnote{Neglecting edge effects due to the DFT circular shift property \cite{oppenheim2001discrete}.}; and in \circled{$2$} the consistency of the covariance estimate $\hR^{(i,j)}[\Delta_{ij}(\up)]$ \cite{de1998weak} (recall that we assume  $\{\tau_{\ell}(\up)\in (T_s\cdot\Zset)\}$, hence so are $\{\Delta_{ij}(\up)\}$).

Therefore, the matrix $\D(\up)$ can be viewed as a collection of all the beamformers created from all possible pairs of the $L$ base stations. Specifically, the diagonal element $D_{\ell\ell}(\up)$ is the auto-beamformer of the $\ell$-th base station based only on $\{\ur_{\ell}[n]\}$, and the off-diagonal element $D_{ij}(\up), i\neq j$, is the cross-beamformer of the $i$-th and $j$-th base stations based only on $\{\ur_{i}[n],\ur_{j}[n]\}$. In light of the above, \eqref{compatformofDPD} can be interpreted as the resulting estimate due to (implicit) ``weighted" joint beamforming.

Another important observation based on \eqref{dataonlyviaestimate} is that $D_{ij}(\up)$ depends on the measured data only via $\hR^{(i,j)}[\Delta_{ij}(\up)]$. Therefore, the set $\{\hR^{(i,j)}[\Delta_{ij}(\up)]:\up\in\Rset^{K_p \times 1}\}_{i,j=1}^L$ is sufficient for the computation of the DPD estimate $\widehat{\up}_{\dpd}$.
\vspace{-0.25cm}

\section{The Proposed Method: One-Bit DPD}\label{sec:OneBitDPDsolution}
In our one-bit problem, only the quantized measurements $\{\uy_{\ell}[n]\}$ are available. However, the key observation specified above (at the end of Subsection \ref{subsec:DPDbeamformer}) essentially means that it suffices to retrieve {\myfontb\emph{only}} the SOS of the observed signals prior to quantization, rather than the signals themselves, $\{\ur_{\ell}[n]\}$.

Fortunately, this can be achieved using the rather simple, yet important extension of the {\emph{arcsine law} \cite{van1966spectrum,jacovitti1994estimation}, as follows.
\begin{corollary}\label{coro1}
	Let $\uxi=\left[\uxi_1^{\tps}\;\uxi_2^{\tps}\right]^{\tps}\in\Cset^{(K_1+K_2)\times1}$ be a zero-mean, circular CN vector with a covariance matrix
	\begin{equation*}
		\Eset\left[\uxi\uxi^{\her}\right]={\begin{bmatrix}
				\Eset\left[\uxi_1\uxi_1^{\her}\right]&\Eset\left[\uxi_1\uxi_2^{\her}\right]\\
				\vspace{-0.35cm}\\
				\Eset\left[\uxi_2\uxi_1^{\her}\right]&\Eset\left[\uxi_2\uxi_2^{\her}\right]\end{bmatrix}}\triangleq{\begin{bmatrix}
				\C_1&\C_{12}\\
				\C_{12}^{\her}&\C_2\end{bmatrix}},
	\end{equation*}
	where $\C_1\in\Cset^{K_1\times K_1}, \C_2\in\Cset^{K_2\times K_2}, \C_{12}\in\Cset^{K_1\times K_2}$. Then,
	\begin{equation}\label{arcsinecross}
		\Eset\left[\calq\left(\uxi_1\right)\calq\left(\uxi_2\right)^{\her}\right]\hspace{-0.05cm}=\hspace{-0.05cm}\frac{2}{\pi}\emph{\sine}^{-1}\left(\mLambda_1^{-\frac{1}{2}}\C_{12}\mLambda_2^{-\frac{1}{2}}\right)\hspace{-0.05cm}\in\Cset^{K_1\times K_2},
	\end{equation}
	where ${\emph{\sine}}^{-1}(z)\triangleq\sin^{-1}(\Re\{z\})+\jmath\cdot\sin^{-1}(\Im\{z\})$ operates elementwise, $\calq(\cdot)$ as in \eqref{defofQ}, and $\mLambda_1$ and $\mLambda_2$ are diagonal matrices with the diagonal elements of $\C_1$ and $\C_2$, resp., on their diagonal. 
\end{corollary}
\begin{proof}
	Applying the arcsine law for CN variables \cite{jacovitti1994estimation} to each of the elements of $\Eset\left[\calq\left(\uxi_1\right)\calq\left(\uxi_2\right)^{\her}\right]$ yields \eqref{arcsinecross}.\hfil\hfil\hfil\hfil\hfil\hfil\hfil\hfil\hfil$\blacksquare$
\end{proof}
\vspace{0.05cm}It follows from Corollary \ref{coro1} that
\begin{equation*}\label{onebitautoandcorrscovariance}
	\Eset\left[\uy_{i}[n]\uy_{j}^{\her}[n-m]\right]=\frac{2}{\pi}\sine^{-1}\left(\bR^{(i,j)}[m]\right),
\end{equation*}
where 
\begin{equation}\label{normalizationdef}
	\bR^{(i,j)}[m]\triangleq\frac{\R^{(i,j)}[m]}{\rho_{i}\rho_{j}},\; \rho^2_{i}\triangleq\frac{|b_i|^2R_s(0)}{M}+\sigma^2_{v_{i}},
\end{equation}
for every $i,j\in\{1,\ldots,L\}$, recalling that the antennas are omni-directional, hence we assume without loss of generality $|a_{\ell,m}(\up)|^2=\tfrac{1}{M}$ for all $m\in\{1,\ldots,M\}$ and $\ell\in\{1,\ldots,L\}$.

Given the $1$-bit quantized signals $\{\uy_{\ell}[n]\}$ from all $L$ base stations, for every candidate $\up$, we may thus construct 
\begin{equation}\label{Rbarhat}
	\widehat{\bR}^{(i,j)}[\Delta_{ij}(\up)]\triangleq\sine\left(\frac{\pi}{2}\cdot\frac{1}{N}\sum_{k=1}^{N}\buy_{i}[k]\buy^{\her}_{j}[k]e^{\jmath\omega_k\Delta_{ij}({\text{\boldmath $p$}})}\right)
\end{equation}
for all pairs $(i,j)$, where $\sine(z)\triangleq\sin(\Re\{z\})+\jmath\cdot\sin(\Im\{z\})$. By virtue of the continuous mapping theorem \cite{mann1943stochastic}, the invertibility of $\sin^{-1}(\cdot)$ on the interval $[-1,1]$, and using \circled{$1$} and \circled{$2$} in the same manner, we conclude that \eqref{Rbarhat} is a consistent estimate of $\bR^{(i,j)}[\Delta_{ij}(\up)]$. Thus, after computing the estimated set $\{\widehat{\bR}^{(i,j)}[\Delta_{ij}(\up)]\}_{i,j=1}^L$, we may further construct (\textit{cf}. \eqref{dataonlyviaestimate})
\begin{equation}\label{Dtildelement}
	\mybar{D}_{ij}(\up)\triangleq N\cdot\ua^{\her}_{i}(\up)\widehat{\bR}^{(i,j)}[\Delta_{ij}(\up)]\ua_{j}(\up),
\end{equation}
for all $i,j\in\{1,\ldots,L\}$, which naturally leads to our proposed OB-DPD estimate, defined as
\begin{equation}\label{compatformofOneBitDPD}
	\hup_{\qdpd}\triangleq\underset{{\text{\boldmath $p$}}\in\Rset^{K_{p}\times 1}}{\argmax}\;\lambda_{\text{max}}\left(\bD(\up)\right).
\end{equation}
We emphasize that although \eqref{compatformofOneBitDPD} is similar in form to \eqref{compatformofDPD}, $\hup_{\qdpd}$ and $\hup_{\dpd}$ are {\myfontb\emph{not}} identical nonetheless. Indeed, the coarse one-bit quantization causes a complete loss of information of the amplitude dimension. However, it turns out that much of the information regarding the position of the transmitter is contained in the relative phases between all signals from all base stations. Hence, much of this information is preserved under $1$-bit quantization. In particular, although $\hup_{\qdpd}$ typically uses significantly less data (/ bits) than $\hup_{\dpd}$, it still enjoys the same asymptotic properties of a ``parallel" DPD estimate based on unquantized data, as shown next.
\vspace{-0.35cm}
\subsection{Asymptotic Properties of One-Bit DPD and Identifiability}\label{subsec:ConsistencyOneBitDPD}
Consider a hypothetical scenario, referred to as Scenario $\mathcal{H}$, which is completely identical to the one under consideration as described in Section \ref{sec:problemformulation}, and in which case $\rho^2_i=\rho^2_j:=\rho^2$ for all $i,j\in\{1,\ldots,L\}$. One (simple) example giving rise to such a scenario is when all $\{b_{\ell}\}$ are identical (e.g., when the base stations are located at equal distances from the transmitter), and all $\{\sigma^2_{v_{\ell}}\}$ are identical. Of course, this is only one possible example out of many. Thus, due to \eqref{normalizationdef}, in this scenario we have  
\begin{equation*}
	\bR^{(i,j)}[m]\triangleq\rho^{-2}\cdot\R^{(i,j)}[m], \; \forall i,j\in\{1,\ldots,L\}.
\end{equation*}
Therefore, in this case \eqref{Rbarhat} converges to the $\rho^{-2}$-scaled covariance matrix of the signals {\myfontb\emph{prior}} to quantization, namely,
\begin{equation*}\label{asymptoticconvergenceonebit}
	\widehat{\bR}^{(i,j)}[\Delta_{ij}(\up)]\xrightarrow[\quad\;\;]{P}\rho^{-2}\cdot\R^{(i,j)}[\Delta_{ij}(\up)].
\end{equation*}
Since multiplication of $\bD(\up)$ by a (positive) scalar is immaterial w.r.t.\ the estimation rule \eqref{compatformofOneBitDPD}---as the scalar multiplication applies equally to all the eigenvalues for any $\up$---it follows that
\begin{equation}\label{asymptoticconvergenceDonebit}
	\bD(\up)\xrightarrow[\quad\;\;]{P}\rho^{-2}\hspace{-0.025cm}\cdot\D(\up) \;\Longrightarrow\; \hup_{\qdpd}\xrightarrow[\quad\;\;]{P}\hup^{(\mathcal{H})}_{\dpd},
\end{equation}
where $\hup^{(\mathcal{H})}_{\dpd}$ denotes the asymptotic DPD estimate based on the unquantized data in Scenario $\mathcal{H}$. In particular, $\hup_{\qdpd}$ inherits some of the asymptotic properties of $\hup^{(\mathcal{H})}_{\dpd}$. For example, if $\hup^{(\mathcal{H})}_{\dpd}$ is consistent w.r.t.\ the SNR and / or sample size, then $\hup_{\qdpd}$ is consistent in the respective sense(s) as well.

Now, consider the general, true scenario, referred to as Scenario $\mathcal{T}$, in which $\{\rho^2_{\ell}\}$ are not necessarily all equal. In this case, $\hup_{\qdpd}$ of Scenario $\mathcal{T}$ would still converge to $\hup^{(\mathcal{H})}_{\dpd}$ of the respective Scenario $\mathcal{H}$, in which the {\myfontb\emph{true}} values of $\{b_{\ell}\}$ and $\{\sigma^2_{v_{\ell}}\}$ are effectively replaced by their respective scaled versions $\{\mybar{b}_{\ell}\triangleq b_{\ell}/\rho_{\ell}\}$ and $\{\mybar{\sigma}^2_{v_{\ell}}\triangleq \sigma^2_{v_{\ell}}/\rho^2_{\ell}\}$. Obviously, $\hup^{(\mathcal{H})}_{\dpd}$ and the DPD estimate of Scenario $\mathcal{T}$, denoted as $\hup^{(\mathcal{T})}_{\dpd}$, would generally differ, and would therefore generally posses different statistical properties. Consequently, according to \eqref{asymptoticconvergenceDonebit}, $\hup_{\qdpd}$ would generally differ from $\hup^{(\mathcal{T})}_{\dpd}$, and in turn they too would generally posses different statistical properties. Nevertheless, $\hup_{\qdpd}$ would still asymptotically converge to $\widehat{\up}^{(\mathcal{H})}_{\dpd}$, its ``parallel" DPD estimate in the quantization-free model, and inherit some of its asymptotic properties, as explained above. In conclusion, OB-DPD and DPD are generally different, as expected. Yet, there always exists a ``parallel" scenario, with the normalized $\{\mybar{b}_{\ell}\}$ and $\{\mybar{\sigma}^2_{v_{\ell}}\}$, and the {\myfontb\emph{same}} emitter position $\up$, in which OB-DPD and DPD asymptotically coincide. 

This outcome is in fact quite intuitive. Indeed, observe that $\sine^{-1}(\cdot)$ is invertible as long as both the real and imaginary parts of its complex-valued argument lie in $[-1,1]$. Moreover, observe that the Signal-to-Noise Ratios (SNRs), as they are expressed in the SOS, are preserved under $1$-bit quantization within the auto-covariance matrix of each base station, since
\begin{equation*}
	\text{SNR}_{\ell}\triangleq\frac{|b_{\ell}|^2R_s(0)}{\sigma^2_{v_{\ell}}}=\frac{|\mybar{b}_{\ell}|^2R_s(0)}{\mybar{\sigma}^2_{v_{\ell}}}, \; \forall \ell\in\{1,\ldots,L\}.
\end{equation*}
Hence, intuitively (and informally), while the {\myfontb\emph{intra}} base station {\myfontb\emph{separated}} SOS information is asymptotically preserved (i.e., for a sufficiently large number of bits), the {\myfontb\emph{inter}} base stations {\myfontb\emph{relative}} amplitude-related SOS information is lost. Thus, $\hup_{\dpd}$ may (implicitly) apply different ``weighting" to each beamformer $D_{ij}(\up)$ based on the unquantized measurements. Conversely, given only $1$-bit measurements from all base stations, $\hup_{\qdpd}$ cannot distinguish, and thus cannot exploit, any differences manifested within the lost amplitude dimension.

It is also interesting to mention that since $\hup_{\qdpd}$ asymptotically coincides with $\hup^{(\mathcal{H})}_{\dpd}$, it follows that the one-bit model in scenario $\mathcal{T}$ is identifiable if and only if the quantization-free model in scenario $\mathcal{H}$ is identifiable. Therefore, the identifiability conditions of our model can be easily deduced from the ``standard" (quantization-free) model (see \cite{weiss2005direct}, Appendix B).  

To summarize, the OB-DPD algorithm is given as follows:
\tcbset{colframe=gray!90!blue,size=small,width=0.49\textwidth,halign=flush center,arc=2mm,outer arc=1mm}
\begin{tcolorbox}[upperbox=visible,colback=white,halign=left]
	\textbf{\underline{The Proposed Solution Algorithm: One-Bit DPD}}\\
	\textbf{Input}: $\left\{\{\uy_{\ell}[n]\}_{n=1}^N\right\}_{\ell=1}^L$, and a $K_p$-dimensional grid. \\ \textbf{Output}: The OB-DPD estimate, $\hup_{{\normalfont \qdpd}}$.
	\begin{enumerate}
			\item For every candidate $\up\in\Rset^{K_p\times 1}$\hspace{0.025cm}:
			\begin{enumerate}
				\item Compute $\{\widehat{\bR}^{(i,j)}[\Delta_{ij}(\up)]\}_{i,j=1}^L$ as in \eqref{Rbarhat};
				\item Compute $\{\mybar{D}_{ij}(\up)\}_{i,j=1}^L$ as in \eqref{Dtildelement};
				\item Compute $\lambda_{\normalfont \text{max}}\left(\bD(\up)\right)$ as in \eqref{compatformofOneBitDPD}.
			\end{enumerate}
			\item Return \eqref{compatformofOneBitDPD}, the maximum point on the grid: $\hup_{\normalfont \qdpd}$.
	\end{enumerate}
\end{tcolorbox}
\section{Simulation Results}\label{sec:simulationresults}
We consider model \eqref{modelequation}, with $L=4$ base stations, located at the corners of a $5\text{[Km]}\times5\text{[Km]}$ square centered about the origin, each equipped with a $M=4$ element uniform linear array. The emitter is located at $\up=[1\;0.5]^{\tps}\text{[Km]}$, transmitting an unknown narrowband signal with a bandwidth of $100\text{[kHz]}$ and a flat spectrum, whose complex envelope is circular CN. Following \cite{weiss2004direct,tirer2015high}, for all $L$ base stations, the channel path-loss attenuation coefficients $\{|b_{\ell}|\}$ were drawn independently from the Gaussian distribution $\mathcal{N}(1,0.1^2)$, the phase of the channel coefficients were drawn independently from the uniform distribution $U[-\pi,\pi)$, and $\{\sigma^2_{v_{\ell}}=\sigma^2_v\}_{\ell=1}^4$. All empirical results are based on $250$ independent trials.

Fig.\ \ref{fig:heatmaps} presents the OB-DPD heat map of a typical realization with an SNR level of $0$[dB] and $N=16$. For comparison, we also present the respective heat map of the classical DPD \cite{weiss2004direct}, based on unquantized\footnote{``unquantized" in the sense of ``up to machine accuracy". In our case, $64$ bits, where the quantization errors are negligible w.r.t.\ the estimation errors.} measurements. As seen, although the heat maps are not identical, the peak corresponding to the OB-DPD estimate is still in close vicinity to the true location.

Fig.\ \ref{fig:RMS} shows the Root Mean Square (RMS) miss distance vs.\ the SNR for a fixed $N=32$ (Fig.\ \ref{fig:RMS_vs_SNR}), and vs.\ the sample size for a fixed $\text{SNR}=0$[dB] (Fig.\ \ref{fig:RMS_vs_N}). Evidently, the accuracy degradation is relatively small, recalling that OB-DPD uses significantly less bits. However, and more importantly, Fig.\ \ref{fig:RMS_vs_N} demonstrates that even when subject to the coarse $1$-bit quantization, for a sufficiently large number of bits, a desired localization accuracy (within the theoretical limitation of the model) can be attained by our OB-DPD estimate.
\begin{figure}[t!]
	\centering
	\includegraphics[width=0.48\textwidth]{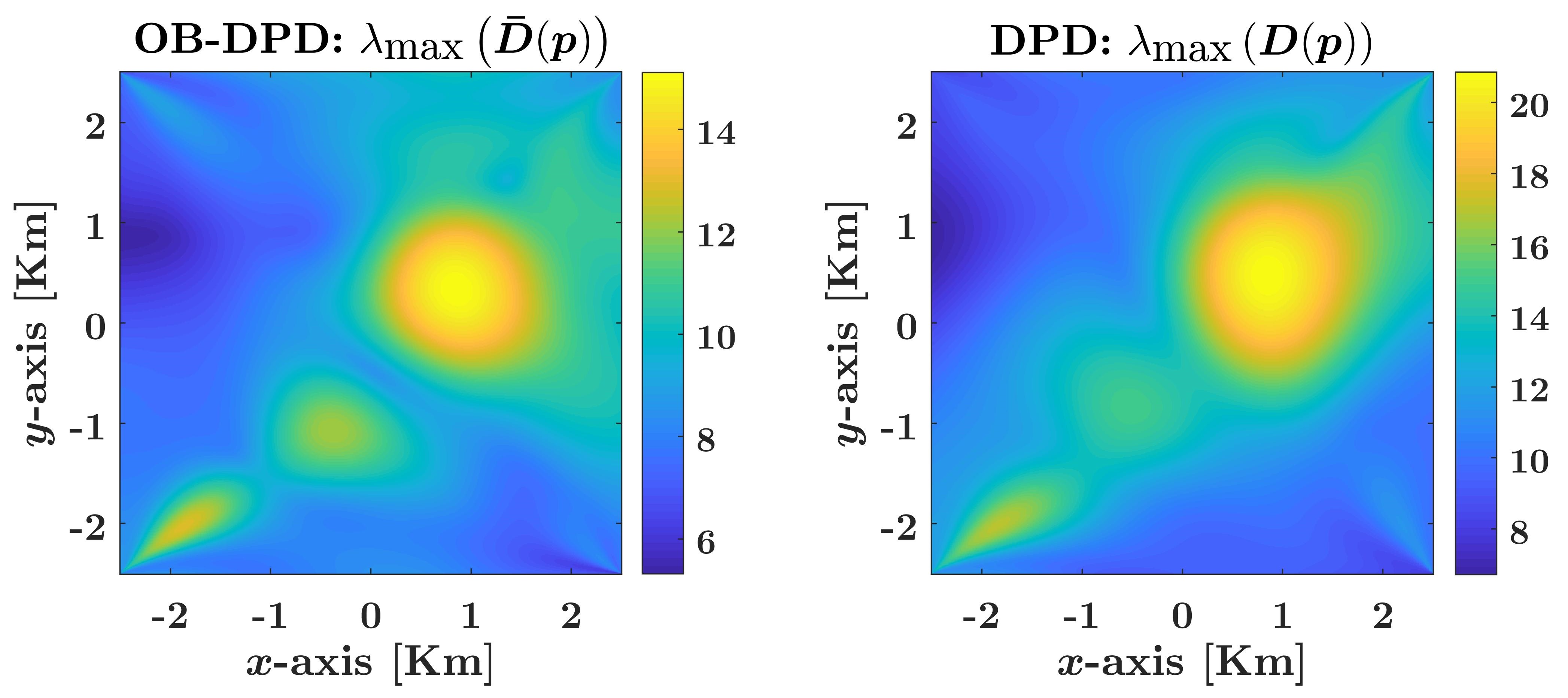}
	\caption{Heat maps of OB-DPD and DPD, $\text{SNR}=0$[dB], $N=16$.}\vspace{-0.3cm}
	\label{fig:heatmaps}
\end{figure}
\begin{figure}[t!]
	\centering
	\begin{subfigure}[b]{0.23\textwidth}
		\includegraphics[width=\textwidth]{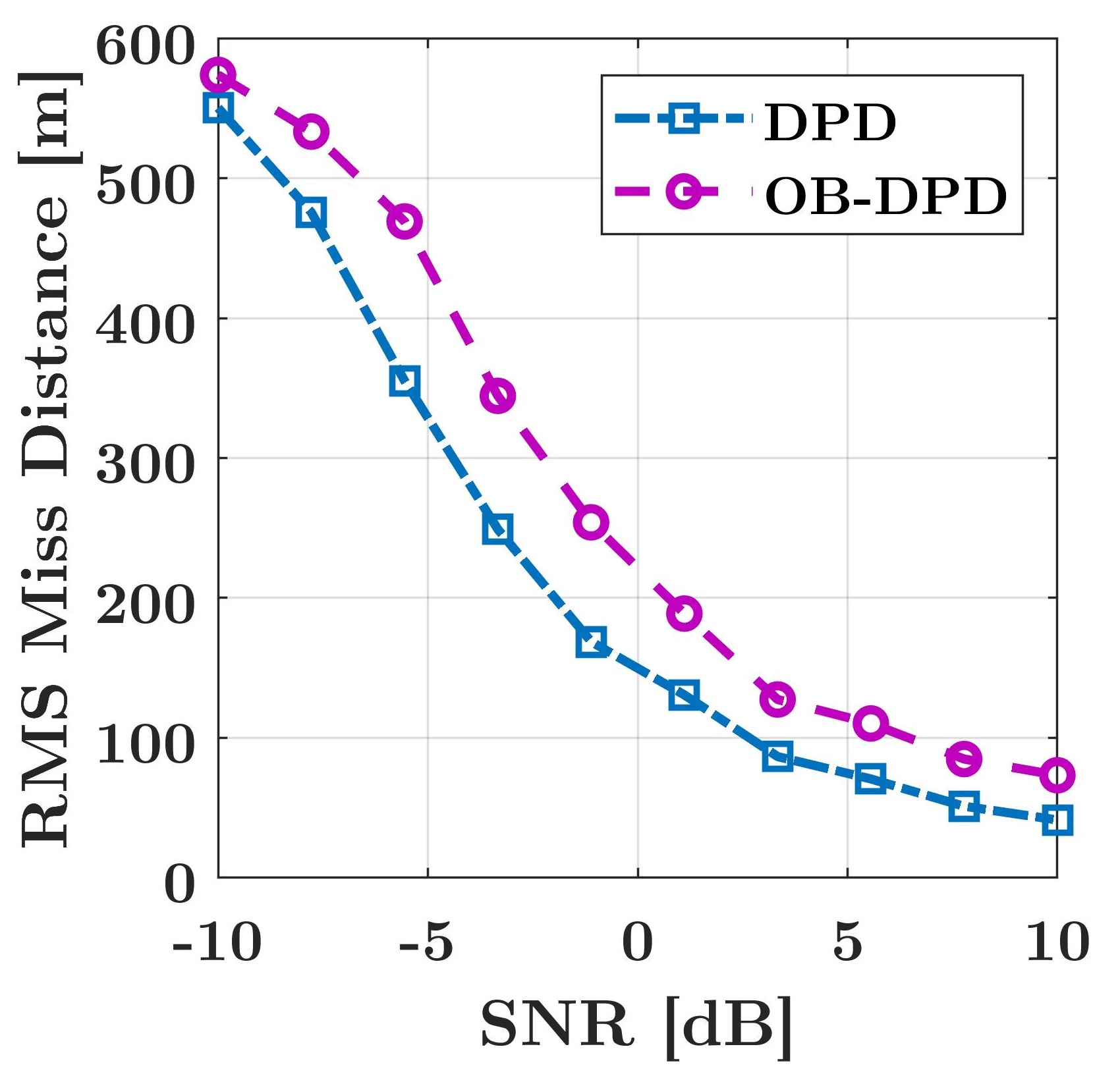}
		\caption{}
		\label{fig:RMS_vs_SNR}
	\end{subfigure}
	~
	\begin{subfigure}[b]{0.23\textwidth}
		\includegraphics[width=\textwidth]{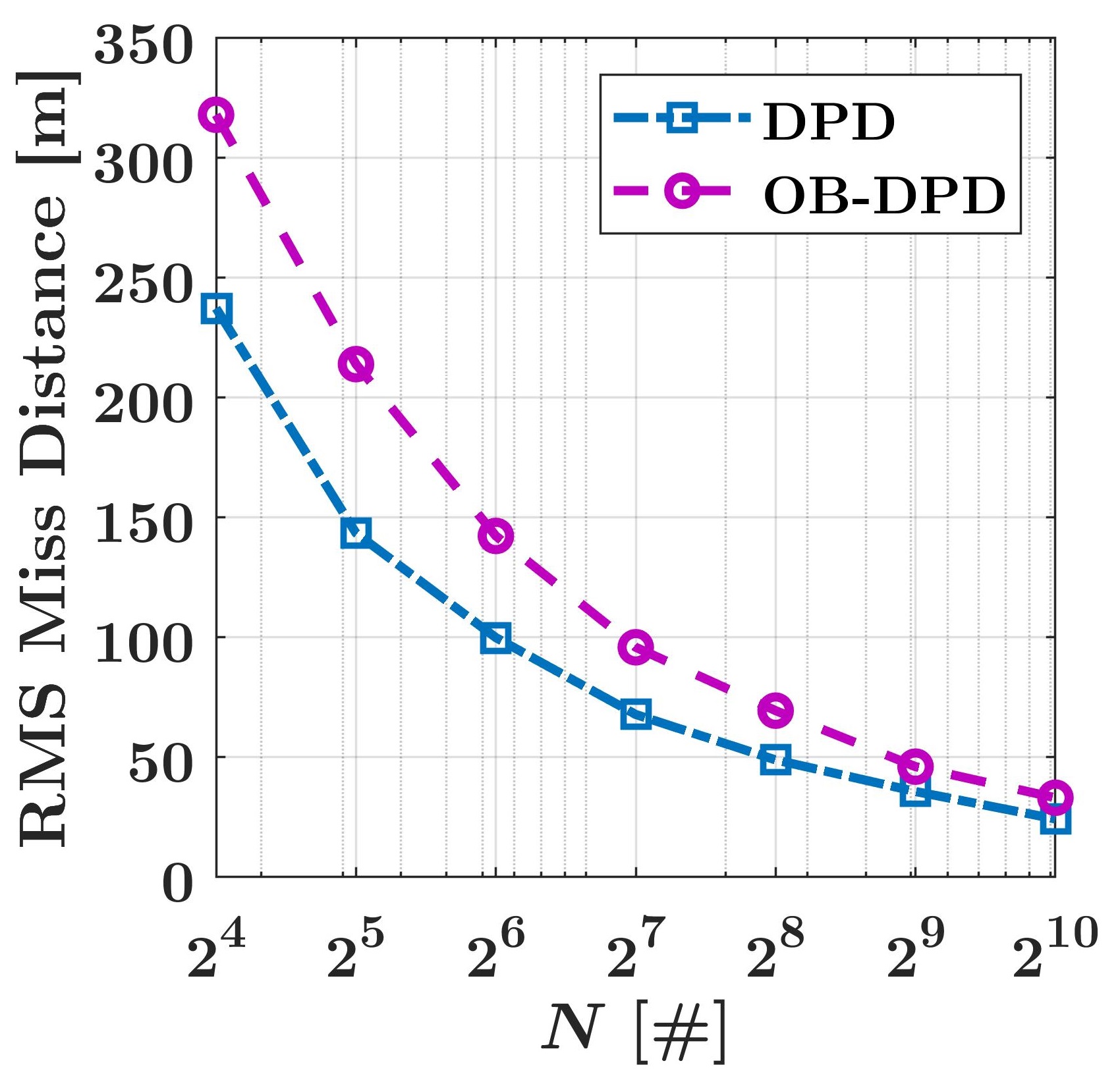}
		\caption{}
		\label{fig:RMS_vs_N}
	\end{subfigure}
	\caption{RMS miss distance (a) vs.\ SNR, $N=32$ (b) vs.\ $N$, $\text{SNR}=0$[dB].}\vspace{-0.5cm}
	\label{fig:RMS}\vspace{-0.1cm}
\end{figure}
\section{Conclusion}\label{sec:conclusion}
We presented a DPD method for one-bit quantized measurements of narrowband Gaussian signals. Based on the partial preservation of the SOS, we proposed the OB-DPD estimate, and established its asymptotic properties. We showed that the underlying signal model is subject to the same identifiability conditions of its respective ``SOS-equivalent" quantization-free model. Relative to DPD without quantization, the requirements on the communication links between the base stations are significantly alleviated by OB-DPD, as the transmitted number of bits from each base station to the central processor is substantially reduced. Potential directions for future research are the quantification of information loss (in terms of estimation accuracy) due to the coarse one-bit quantization, extensions to multiple emitters, and to the more general case of non-Gaussian signals. Furthermore, an analytical analysis (approximately) predicting the minimal number of bits required from each station for a desired localization accuracy level, would be both instructive from a theoretical point of view, and beneficial from a practical point of view.

{\footnotesize{\section{Acknowledgment}\label{sec:acknowledgment}
\vspace{-0.2cm}
This work was supported, in part, by NSF under Grant No.\ CCF-1717610 and ONR under Grant No.\ N00014-19-1-2665.}}
\vspace{-0.2cm}
\bibliographystyle{IEEEbib}
\bibliography{refs}

\end{document}